\documentclass{article}

\usepackage{amssymb}
\usepackage{amsthm}

\usepackage[latin1]{inputenc}  

\usepackage{amsmath}
\usepackage{amsthm}

\newtheorem{thm}{Theorem}
\newtheorem{lem}[thm]{Lemma}
\newcommand{\Pro}[1]{\mbox{{\rm Pr}}\left(#1\right)}

\title{Quasi-Random Rumor Spreading: Reducing Randomness Can Be Costly}

\author{Benjamin Doerr \and  Mahmoud Fouz}
\begin{document}

\maketitle
%\begin{frontmatter}



%\address[Doerr]{Max-Planck-Institut f\"ur Informatik\\
%}

%\address[Fouz]{Universit\"at des Saarlandes\\
%	D-66123 Saarbr\"ucken, Germany}

\begin{abstract}
We give a time-randomness tradeoff for the quasi-random rumor spreading protocol proposed by Doerr, Friedrich and Sauerwald [SODA 2008] on complete graphs. In this protocol, the goal is to spread a piece of information originating from one vertex throughout the network. Each vertex is assumed to have a (cyclic) list of its neighbors. Once a vertex is informed by one of its neighbors, it chooses a position in its list uniformly at random and then informs its neighbors starting from that position and proceeding in order of the list. Angelopoulos, Doerr, Huber and Panagiotou [Electron.~J.~Combin.~2009] showed that after $(1+o(1))(\log_2 n + \ln n)$ rounds, the rumor will have been broadcasted to all nodes with probability $1 - o(1)$. 

We study the broadcast time when the amount of randomness available at each node is reduced in natural way. In particular, we prove that if each node can only make its initial random selection from every $\ell$-th node on its list, then there exists lists such that $(1-\varepsilon) (\log_2 n + \ln n - \log_2 \ell - \ln \ell)+\ell-1$ steps are needed to inform every vertex with probability at least $1-O\bigl(\exp\bigl(-\frac{n^\varepsilon}{2\ln n}\bigr)\bigr)$. This shows that a further reduction of the amount of randomness used in a simple quasi-random protocol comes at a loss of efficiency.
\end{abstract}


\section{Introduction}

%We give a tradeoff between running time and randomness for the \textit{}
\subsection{Randomized Rumor Spreading}

We consider the \textit{rumor spreading problem}, i.e., the problem of dissemanting information on networks: given a graph $G$ and a node $v$ that has some piece of information, the goal is to spread this piece of information to all nodes, where in each step only adjacent nodes can communicate with each other. A simple randomized algorithm for this problem is for each informed node to select, in each iteration, one of its neighbors uniformly at random and then to send the piece of information to that neighbor. In case of the complete graph, Frieze and Grimmett \cite{FG85} showed that $(1+o(1)) (\log_2 n + \ln n)$ iterations are sufficient in order to inform every node with probability $1-o(1)$. This was tightened by Pittel \cite{Pi87} who proved that $\log_2 n + \ln n + O(1)$ iterations are sufficient for that. 

Note that each node needs $\lceil\log_2 (n-1)\rceil$ random bits in order to choose one of its neighbors uniformly at random. Since most nodes keep informing for $\Omega(\log n)$ rounds until all nodes are informed, a node will use $\Omega(\log^{2} n)$ bits on average with probability $1-o(1)$. Recently, Doerr et al.~\cite{BC} reduced the amount of randomness needed for each node to $O(\log n)$ while maintaining a logarithmic running time. In their \textit{quasi-random model}, they assume that every node has a (cyclic) list of its neighbors. This list dictates the order in which the node informs its neighbors. Once a node $v$ gets informed, it selects a position in its list uniformly at random and proceeds by informing all nodes starting from this position. In other words, after an initial random choice that requires $\lceil\log_2 (n-1)\rceil$ random bits, the node proceeds deterministically and needs no further random bits. %The running time of this protocol was further tightened by Angelopoulos et al \cite{ABDFHP09}. 

In this paper, we complement this effort of reducing the amount of randomness by providing a tight time-randomness tradeoff. Whereas the reduction of random bits from $O(\log^2 n)$ to $O(\log n)$ at each vertex comes at no loss of efficiency, we show that a subsequent reduction of randomness in a more general model will incur additional rounds for particular choices of the lists. In this \textit{gate model}, we assume that every vertex makes its random choice only from a subset of special vertices equidistantly distributed on its list. Roughly speaking, we prove that if $\ell$ is the distance between two gates, then, with probability $1-o(1)$, $\ell$ additional rounds are needed to inform all vertices. %We refer the reader to Section \ref{sec:limitedRandomness} for the precise statements. 

\subsection{The Dilemma of Randomization}

Probabilistic methods have given rise to a large number of algorithms that utilize random choices to perform difficult tasks efficiently. Often, these probabilistic algorithms beat deterministic algorithms not only in terms of running time, but also in terms of complexity, or rather simplicity. On the downside, probabilistic algorithms have two major drawbacks: First, it is highly non-trivial to produce truly random bits. Second, although these algorithms perform quite well in expectation or even with high probability, there is no guarantee that they will always do so. Derandomized versions of these algorithms are therefore highly desirable. Unfortunately, it remains one of the big open questions in computer science whether it is always possible to completely derandomize polynomial time randomized algorithms without sacrificing the polynomial running time. For recent developments on this question, we refer the interested reader to a survey by Impagliazzo \cite{I07}. There are two ways to work around this problem. First, one can try to reduce the amount of randomness needed in these algorithms without any (or minor) sacrifices in terms of efficiency. Second, one can can study the relationship between the running-time and the amount of randomness used. Both approaches have been applied to several problems (see, e.g., \cite{B91,CG89,KR88,PU90,BC}). In this paper, we apply the second approach to the rumor spreading problem.  %Krizanc, Peleg and Upfal \cite{krizanc88} proposed  

\section{Time-Randomness Tradeoff}

We consider a generalization of the quasi-random model, where the number of available random bits at each vertex is less than $\log_2 n$.
%We prove the lower bound in Theorem~\ref{thm1} in a more general model that takes into account the case of having less random bits than needed to choose any neighbor uniformly at random.
More precisely, in the \textit{gate model} we assume that every vertex makes his random choice only from a subset of special vertices among his neighbors. These \textit{gates} are equidistantly distributed in the list of each node at distance $\ell \leq  n-1 $ from each other, starting from the first neighbor in the list. Since the number of random bits needed decreases when $\ell$ increases, we can think of $\ell$ as a randomness measure. After the initial random choice of a gate neighbor, the vertex continues to inform all the subsequent neighbors as before. Note that for $\ell = 1$, the gate model reduces to the standard quasi-random model. For clarity, we assume that $n/\ell$ is integral.
\begin{thm}
\label{lowerbound_gatemodel}
There exist lists such that the quasi-random gate model with randomness parameter $\ell \in [n]$ on the complete graph on $n$ vertices needs at least
	\begin{equation*}
	  (1-\varepsilon) (\log_2 n + \ln n - \log_2 \ell - \ln \ell)+\ell-1 
	\end{equation*}
steps to inform every vertex with probability $1-O\bigl(\exp(-\frac{n^\varepsilon}{2\ln n})\bigr)$.
\end{thm}
%Note that for $\epsilon := \frac{2\ln \ln n}{\ln n}$, the error probability is at most $O(n^{-1})$. 
Theorem \ref{lowerbound_gatemodel} gives a natural tradeoff between the amount of randomness used and the broadcast time needed. Note that such a result cannot hold for arbitrary lists. In particular, for randomly chosen lists the starting point does not matter. So even if all lists start informing from the first node on their list, the process amounts exactly to the classical quasi-random model for which Angelopoulos et al.~proved the following lower bound.
\begin{thm}[Angelopoulos et al.~\cite{ADHP09}]
 For all lists, the quasi-random protocol on the complete graph on $n$ vertices informs all vertices in 
\begin{equation*}
 (1+o(1))(\log_2 n + \ln n)
\end{equation*}
steps with probability $1-o(1)$.
\end{thm}

%the process essentially amounts to the random model where each informed node informs another node chosen uniformly at random. The only difference is that a node never informs the same node twice, which makes the process at least as fast as the random model that has a broadcast time of $(1+o(1))(\log_2 n + \ln n)$ (see \cite{FG85}).    

The proof simulates a process consisting of two phases that finishes no later than the actual model. The second phase of the process can be reduced to the following problem.  Let $e_1,\dots, e_n$ be a sequence of $n$ elements. Assume that $m$ elements are already marked. In addition, we mark $i$ elements uniformly at random with replacement. The following lemma shows that there is a large interval of unmarked elements with high probability for a reasonable choice of $m$ and $i$. 

\begin{lem}
\label{freeinterval_lemma}
   Let $e_1,\dots, e_n$ be a sequence of $n$ elements out of which $m$ elements are `pre-marked' and, furthermore, $i \in \omega\bigl(\ln^2 n\bigr)$ elements are marked uniformly at random with replacement. Then, for all $\varepsilon >0$ and large $n$, the largest interval of unmarked elements has length at least 
\begin{equation*}k= \tfrac{n}{i}(1-\epsilon)\ln n\end{equation*}
 with probability at least 
%\begin{equation*}
$1-\exp\Bigl(-\tfrac{1}{2}\bigl(n^{\varepsilon}/k +mn^{-1+\varepsilon}\bigr)\Bigr)$.
%\end{equation*}
\end{lem}\begin{proof}
We partition the sequence into disjoint intervals of length $k$. We have $\frac{n}{k}$ such intervals. We call an interval \emph{marked} if it contains at least one marked element. At most $m$ of these intervals contain a previously (deterministically) marked element. For any other interval $I$, we have
\begin{equation}
\label{freeprob}
 %\Pro{I \text{ is marked}} &\leq 1-(1-\tfrac{k}{n-1})^{i-k}(1-\tfrac{k-1}{n-1}).
\Pro{I \text{ is marked}} = 1-\bigl(1-\tfrac{k}{n}\bigr)^{i}.
\end{equation}
Note that these intervals are not marked independently. However, the fact that some of these intervals $I_1, \dots, I_j$ are marked by the random process implies that there are at most $i-j$ random selections left that could lead to the marking of another interval $J$ since all intervals are disjoint. Thus, the events that intervals are marked are \emph{negatively correlated}: if some intervals are marked, the probability that another one is also marked cannot increase, i.e., 
\begin{equation}
\label{negative_correlation}
\begin{split}
\Pro{\text{$I$ is marked} \mid I_1, \dots, I_j \text{ are marked}} \\
			\leq \Pro{\text{$I$ is marked}}. 
\end{split}
\end{equation} 
%Using this fact, a simple calculation (see appendix) yields the desired result.
Let $I_1, \dots, I_{n/k}$ denote the intervals. By a slight abuse of notation, we also denote by $I_j$ the event that interval $I_j$ is marked.

We will need the following fact to complete the proof: for $x \leq \tfrac{1}{2}$, we have 
\begin{equation}
\label{exp_ineq}
 1-x \geq e^{-x-x^2}.
\end{equation}

We compute,
\begin{alignat}{2}
&\Pro{\text{all intervals are marked}} \\
 & = \Pro{\bigwedge_{1\leq j \leq n/k} I_j}\notag\\
				      \begin{split}&= \Pro{I_1} \cdot \Pro{I_2 \mid I_1} \cdot \Pro{I_3 \mid I_1 \wedge I_2 } \notag\\
							& \quad \quad \quad  \quad \cdots \Pro{I_{n/k} \mid I_1 \wedge \cdots \wedge I_{n/k-1}}\end{split}\notag\\
				      &\leq \prod_{1 \leq j\leq n/k} \Pro{I_j} & \text{by \eqref{negative_correlation}} \notag\\
					&\leq \Bigl(1-(1-\tfrac{k}{n})^{i}\Bigr)^{\frac{n}{k}-m}\notag\\
					&\leq \Bigl(1-\exp\bigl(i(-\tfrac{k}{n}-\tfrac{k^2}{n^2})\bigr)\Bigr)^{\frac{n}{k}-m} & \text{by \eqref{exp_ineq}}\notag\\
					&\leq \Bigl(1-n^{-1+\epsilon} n^{-\tfrac{(1-\epsilon)^2\ln n}{i}}\Bigr)^{\frac{n}{k}-m}  \label{def_k}\\
					&\leq \Bigl(1-\tfrac{1}{2}n^{-1+\epsilon}\Bigr)^{\frac{n}{k}-m} \label{def_i}\\
					&\leq \exp\Bigl(-\tfrac{1}{2}n^{-1+\epsilon}(\tfrac{n}{k}-m)\Bigr) \label{exp_ineq2}.						
\end{alignat}

Here, \eqref{def_k} follows from the definition of $k = \frac{n}{i}(1-\varepsilon) \ln n$, \eqref{def_i} follows, for large $n$, from the assumption that $i \in \omega\bigl(\ln^2 n\bigr)$, and \eqref{exp_ineq2} follows from $1+x \leq e^x$. 
\end{proof}
With those facts at hand we now prove Theorem \ref{lowerbound_gatemodel}.

\begin{proof}[Proof of Theorem \ref{lowerbound_gatemodel}.]
Assume that all vertices have (almost) the same list $[1,2,\dots,n]$, except that each vertex is excluded from its own list. As a result the nodes do not have exactly the same gates. However, the $i$-th gate of any list will be either node $(i-1)\ell+1$ or node $(i-1)\ell+2$. We will therefore treat both as essentially the same node, i.e, whenever the $i$-th gate of any node is informed, the $i$-th gate of every other node is also informed immediately.  
%However, the gates of two nodes can differ by at most one position per gate. 
%\footnote{Strictly speaking, each node should exclude itself from its list. It can be easily argued that this makes no difference, so we refrain from doing so in the proof.}. 
%Thus, the gates of all nodes are $\{1,2, \quad \ell+1, \ell+2, \quad 2\ell+1, 2\ell +2, \dots,n-\ell+1, n-\ell+2\}$, where two consecutive vertices stand for essentially the same gate. We will therefore treat every such pair as one gate, that is, if one of these vertices is informed, its `partner' vertex is also considered as informed immediately. 
Clearly, this assumption only speeds up the process.
%Then the gates are defined to be $\{1,\ell+1, 2\cdot \ell+1,\dots,n-\ell+1\}$. 
We now bound from below the time needed to inform all vertices by a process consisting of two phases that finishes at least as early as the actual rumor spreading model.

In the first phase, which lasts for $(1-\varepsilon) \log_2 (n/\ell)$ steps, we only assume that the number of informed vertices doubles in each step. Note that this is optimal since in each step the number of informed vertices can at most double. So we end up with at most $(\frac{n}{\ell})^{1-\varepsilon}$ informed gates. We shall not use any further information on how these gates became informed. 

In the second phase, we assume that every vertex is allowed to spread the rumor even if it has not received it yet. In other words, we bring forward the random choice of each vertex that has not yet started to spread the rumor. This modification will only speed up the process. In particular, at the beginning of the second phase, every such vertex chooses one of the gates uniformly at random and then spreads the rumor accordingly. We will prove that even under this assumption, we additionally need $(1-\varepsilon)\ln (n/\ell) + \ell-1$ steps until every vertex has received the rumor. 

Using Lemma \ref{freeinterval_lemma}, we argue that after the random choice of all these vertices, there is a large interval of uninformed gates. Let $n_0$ denote the number of such vertices.  Note that now the length of the sequence is $n/\ell$. So by Lemma \ref{freeinterval_lemma} with $i = n_0$ and $m = (\frac{n}{\ell})^{1-\varepsilon}$, there is such a free interval of length $k = \frac{n}{n_0 \ell}(1-\varepsilon)\ln (n/\ell) \geq \tfrac{1-\varepsilon}{\ell}\ln (n/\ell)$
with probability at least 
\begin{align*}
&1-\exp\bigl(-\tfrac{1}{2}(n/\ell)^{\varepsilon}/k +m(n/\ell)^{-1+\varepsilon}\bigr) \\
&= 1-\exp\bigl(-\tfrac{1}{2}(n/\ell)^{\varepsilon}/k +1\bigr)\\
&\geq 1-\exp\bigl(-\tfrac{n^{\varepsilon}}{2\ln n} +1\bigr).
\end{align*}
We need at least $\ell-1$ steps to reach this interval and additionally $\ell \cdot k \geq (1-\varepsilon)\ln (n/\ell)$ steps to inform all vertices in this interval. So in total, we need
\begin{equation}
 (1-\varepsilon) (\log_2 n + \ln n - \log_2 \ell - \ln \ell)+\ell-1
\end{equation}
steps in order to inform every vertex with probability at least $1-O\bigl(\exp\bigl(-\frac{n^{\varepsilon}}{2\ln n}\bigr)\bigr)$.%Setting $\epsilon := \frac{\ln 2 + 2\ln \ln n}{\ln n}$ yields the result.
\end{proof}
% %\bibitem{ABDFHP09}
% %Angelopoulos, S., M.~Bl\"{a}ser, B.~Doerr, M.~Fouz, A.~Huber, and K.~Panagiotou.
% %\newblock \textit{{T}ight {B}ounds for {Q}uasirandom {R}umour {S}preading.}
% %\newblock {submitted}.
% \bibitem{B91}
% Bach, E.
% \newblock \textit{Realistic analysis of some randomized algorithms.}
% \newblock {J. Comput. Syst. Sci.}, \textbf{42} (1991):30--53.
% 
% \bibitem{CG89}
% Chor, B., and O.~Goldreich.
% \newblock \textit{On the power of two-point based sampling.}
% \newblock {J. Complex.}, \textbf{5} (1989):96--106.
% 
% \bibitem{BC}
% Doerr, B., T.~Friedrich, and T.~Sauerwald.
% \newblock \textit{Quasirandom rumor spreading.}
% \newblock In {Proc.~of the 19th Annual ACM-SIAM Symp.~on Disc.~Alg.}, pages 773--781, 2008.
% 
% \bibitem{FG85}
% Frieze, A.M., and G.R. Grimmett.
% \newblock \textit{The shortest-path problem for graphs with random arc-lengths.}
% \newblock {Discr. Appl. Math.}, \textbf{10} (1985):57--77.
% 
% \bibitem{KR88}
% Karloff, H., and P.~Raghavan.
% \newblock \textit{{R}andomized {A}lgorithms and {P}seudorandom {N}umbers.}
% \newblock In {Proc.~of 20th Annual ACM Symposium on Theory of Comp.}, 1988. 
% 
% \bibitem{KPU88}
% Krizanc, D., D.~Peleg, and E.~Upfal.
% \newblock \textit{{A} {T}ime-{R}andomness {T}radeoff for {O}blivious {R}outing.}
% \newblock In {Proc.~of 20th Annual ACM Symposium on Theory of Comp.}, 1988.
% 
% \bibitem{Pi87}
% Pittel, B.
% \newblock \textit{On spreading a rumor.}
% \newblock {SIAM Jour.~on Appl.~Math.}, \textbf{47} (1987):213--223.
% 
% \end{thebibliography}
\bibliographystyle{plain}
\bibliography{soda09broadcast.bib}
\end{document}